\documentclass[%
 reprint,
superscriptaddress,
nofootinbib,
 amsmath,amssymb,
 aps,
floatfix,
]{revtex4-2}

\usepackage{amsthm}
\usepackage{tikz}
\usetikzlibrary{decorations.pathmorphing}
\tikzset{snake it/.style={decorate, decoration=snake}}
\newtheorem{theorem}{Theorem}[section]

\newtheorem{lemma}[theorem]{Lemma}

\usepackage{braket}
\usepackage{bbm}
\normalsize
\newcommand{\defn}{\mathrel{\mathop:}=} 

\newcommand{\nb}[1]{\color{blue}}

\newcommand{\hl}[1]{\color{magenta}}

\usepackage[
	 pagebackref=false,
	 colorlinks=true,
 linkcolor=blue,
 urlcolor=blue,
 filecolor=black,
 citecolor=blue,
 pdfstartview=FitV,
 pdftitle={},
 pdfauthor={},
 pdfsubject={},
 pdfkeywords={},
 pdfpagemode=None,
 bookmarksopen=true
 ]{hyperref}

\usepackage[normalem]{ulem}
\usepackage{amsmath}
\usepackage{enumerate}
\usepackage{amsfonts}
\usepackage{epsfig}
\usepackage{mathbbol}

\def\Tr{\mathop{\rm Tr}}
\def\tr{\mathop{\rm tr}}

\newcommand\vev[1]{{\ensuremath{\left\langle{#1}\right\rangle}}}

\newcommand{\beginsupplement}{%
    \setcounter{table}{0}
    \renewcommand{\thetable}{S\arabic{table}}%
    \setcounter{figure}{0}
    \renewcommand{\thefigure}{S\arabic{figure}}%
   }
\newcommand{\be}{\begin{equation}}
\newcommand{\ee}{\end{equation}}
\newcommand{\bea}{\begin{eqnarray}}
\newcommand{\eea}{\end{eqnarray}}
\newcommand{\bega}{\begin{gather}}
\newcommand{\eega}{\end{gather}}

\newcommand{\bi}{\begin{itemize}}
\newcommand{\ei}{\end{itemize}}
\newcommand{\ben}{\begin{enumerate}}
\newcommand{\een}{\end{enumerate}}
\newcommand{\bca}{\begin{cases}}
\newcommand{\eca}{\end{cases}}
\newcommand{\bln}{\begin{align}}
\newcommand{\eln}{\end{align}}
\newcommand{\bst}{\begin{split}}
\newcommand{\est}{\end{split}}
\def\ie{\begin{equation}\begin{aligned}}
\def\fe{\end{aligned}\end{equation}}
\newcommand{\bma}{\le(\begin{matrix}}
\newcommand{\ema}{\end{matrix}\ri)}

\DeclareMathOperator{\sech}{sech}
\newcommand{\alg}{\mathcal{A}}

\def\le{\left}
\def\ri{\right}

\newcommand{\beq}{\begin{equation}}
\newcommand{\eeq}{\end{equation}}

\newcommand{\hs}{\mathcal{H}}
\newcommand{\op}{\mathcal{O}}
\newcommand{\ttwo}{\text{II}}

\newcommand{\vn}{\textrm{vN}}

\usepackage{anyfontsize}

\begin{document}

\title{
A covariant regulator for entanglement entropy: \\ proofs of the Bekenstein bound and QNEC}

\preprint{MIT-CTP/5332}

\author{Jonah Kudler-Flam}
\affiliation{School of Natural Sciences, Institute for Advanced Study, Princeton, NJ 08540, USA}
\affiliation{Princeton Center for Theoretical Science, Princeton University, Princeton, NJ 08544, USA}
\author{Samuel Leutheusser}
\affiliation{Princeton Gravity Initiative, Princeton University, Princeton NJ 08544, USA} 
\author{Adel A.~Rahman}
\affiliation{Stanford Institute for Theoretical Physics and Department of Physics, Stanford University, Stanford, CA 94305-4060, USA}
\author{Gautam Satishchandran}%
\affiliation{Princeton Gravity Initiative, Princeton University, Princeton NJ 08544, USA} 
\author{Antony J.~Speranza}
\affiliation{Department of Physics, University of Illinois, Urbana-Champaign, Urbana IL 61801, USA}

\begin{abstract}
While von Neumann entropies for subregions in quantum field theory universally contain ultraviolet divergences,  differences between von Neumann entropies 
are finite and well-defined in many physically relevant scenarios. 
We demonstrate that such a notion of entropy differences can be rigorously defined in
quantum field theory in a general curved spacetime by introducing a novel, covariant regulator for the entropy based on
the modular crossed product. This regulator associates a type II von Neumann algebra
to each spacetime subregion, resulting in  well-defined renormalized entropies.  
This prescription reproduces formulas for entropy differences that coincide with heuristic formulas widely used in the literature, and we prove that it satisfies desirable properties such as unitary invariance and concavity. As an application, we provide proofs of the Bekenstein bound and the quantum null energy condition, formulated directly in terms of vacuum-subtracted von Neumann entropies.
\noindent 
 
\end{abstract}

\date{December 12, 2023}
\maketitle

Entanglement entropy is a widely studied quantity in physics and has played a central role in the theory of quantum information, the characterization of phases of matter, quantum field theories, and quantum gravity. 
In standard treatments of entanglement entropy, one assumes a decomposition
 of the Hilbert space 
into tensor 
factors associated with subsystems $A$ and $B$, $\mathcal{H} = \mathcal{H}_A\otimes
\mathcal{H}_B$.  Global states on $\mathcal{H}$ give rise to density matrices on the 
subsystem $A$ upon tracing over the Hilbert space $\mathcal{H}_B$, and 
the resulting state on $A$ is generically mixed if the global state was entangled
between $A$ and $B$.  The entropy of this reduced density matrix provides a measure 
of correlations between the two subsystems.  

This procedure carries over to quantum field theory after regulating,
for example, by discretizing the theory on a lattice.  
States for the lattice theory  describe the field configuration
at  an instant in time, and the lattice spacing $\lambda$ provides an ultraviolet 
(UV) cutoff that allows the Hilbert space to factorize into spatial
subregions.  Due to the high degree of vacuum entanglement 
between modes at short distances (see figure \ref{fig:S_div}), the entanglement
entropy of any state with a good continuum limit is {\it universally divergent}.
For a density matrix $\rho_\lambda$ associated with the spatial region
$\Sigma_{\mathcal{R}}$, the entropy has an area law divergence 
in the limit $\lambda\rightarrow 0$ \cite{sorkin1983entropy}
\begin{align}
  \label{eq:div}
  S_{\textrm{vN}}(\rho_{\lambda}) \:= -\tr \rho_{\lambda} \log \rho_{\lambda} \propto \frac{A}{\lambda^{d-2}} + \dots,
\end{align}
where $d$ is the spacetime dimension and $A$ is the area of the entangling 
surface.\footnote{The divergence is logarithmic in $d=2$.} 

\begin{figure}
    \centering
 \begin{tikzpicture}[scale=1]
  \tikzset{decoration={snake,amplitude=.4mm,segment length=2mm,post length=0mm,pre length=0mm}}
\colorlet{myred}{red!80!black}
\colorlet{myblue}{blue!80!black}
\colorlet{mygreen}{green!80!black}
\colorlet{mydarkred}{red!50!black}
\colorlet{mydarkblue}{blue!50!black}
\colorlet{mylightblue}{mydarkblue!6}
\colorlet{myvlightblue}{mydarkblue!3}
\colorlet{mypurple}{blue!40!red!80!black}
\colorlet{mydarkpurple}{blue!40!red!50!black}
\colorlet{mylightpurple}{mydarkpurple!80!red!6}
\colorlet{myorange}{orange!40!yellow!95!black}
\draw[black, very thick] (-3,-2) rectangle (3,2);
\fill[mylightblue] (-3,-2) rectangle (3,2);
\draw[black, very thick] (0,0) circle (1.25);
\fill[mylightpurple, very thick] (0,0) circle (1.25);
\draw[red, thick,decorate](1.05,0)--(1.45,0);
\draw[red, thick, rotate = 45/2, snake it](1.05,0)--(1.45,0);
\draw[red, thick, rotate = 45, snake it](1.05,0)--(1.45,0);
\draw[red, thick, rotate = 45*3/2, snake it](1.05,0)--(1.45,0);
\draw[red, thick, rotate = 90, snake it](1.05,0)--(1.45,0);
\draw[red, thick, rotate = 45*5/2, snake it](1.05,0)--(1.45,0);
\draw[red, thick,rotate = 135, snake it](1.05,0)--(1.45,0);
\draw[red, thick, rotate = 45*7/2, snake it](1.05,0)--(1.45,0);
\draw[red, thick,rotate = 180, snake it](1.05,0)--(1.45,0);
\draw[red, thick, rotate = 45*9/2, snake it](1.05,0)--(1.45,0);
\draw[red, thick,rotate = 225, snake it](1.05,0)--(1.45,0);
\draw[red, thick, rotate = 45*11/2, snake it](1.05,0)--(1.45,0);
\draw[red, thick,rotate = 270, snake it](1.05,0)--(1.45,0);
\draw[red, thick, rotate = 45*13/2, snake it](1.05,0)--(1.45,0);
\draw[red, thick,rotate = 315, snake it](1.05,0)--(1.45,0);
\draw[red, thick, rotate = 45*15/2, snake it](1.05,0)--(1.45,0);
\node[scale = 1.5] at (0,0) {$\Sigma_{\mathcal{R}}$};
\end{tikzpicture} 
    \caption{The entanglement entropy of a causally complete spacetime subregion $\mathcal{R}$ is divergent due to high-energy modes (red) across the entangling surface. Here $\Sigma_{\mathcal{R}}$ is a Cauchy surface for $\mathcal{R}$ and its boundary is the entangling surface.}
    \label{fig:S_div}
\end{figure}
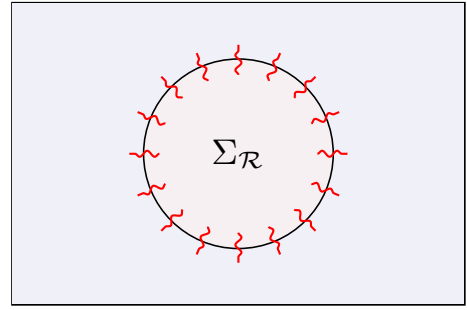

While lattices provide a practical means for investigating information 
theoretic
aspects of quantum field theories, there are several disadvantages.
 The lattice structure drastically changes the Hilbert space of the quantum theory, and it can be quite subtle (if at all possible) to determine how states and operators should be mapped to the continuum.\footnote{Indeed, there exist quantum field theories where lattice regulators do not exist \cite{nielsen1981absence,nielsen1981absenceII,nielsen1981no,friedan1982proof,2021arXiv210103320H}.} 
Furthermore, the lattice does not respect the underlying smooth structure
of the manifold  
on 
which the quantum fields live, causing quantities such as the entropy --- 
which should 
depend covariantly on the geometry of the subregion ---
 to have 
spurious sensitivity to the
short distance details of the lattice.    The need to define the density matrix 
on a spatial slice is another drawback of the lattice for Lorentz-invariant 
theories, since it obscures the underlying causal properties of the local operators 
in the theory, which are only manifest in a spacetime description.  

This motivates using a formulation in which spacetime
causality and covariance is manifest.  The Haag-Kastler algebraic approach to
 quantum field
theory provides such a formulation \cite{haag1964algebraic,haag2012local}.  
In this approach, causally complete spacetime regions, rather than spatial slices,
 take on a central role,
and each region, $\mathcal{R}$, is associated with a complete algebra of local operators, $\mathcal{A}$, comprising 
a subsystem in the theory.  However, algebraic QFT is not without its own complications.
The algebras of interest are Type $\text{III}_1$ von Neumann algebras,
for which there are no natural Hilbert space factorizations nor well-defined notions
of local density matrices.  The physical manifestation of this statement is that 
entanglement entropies are infinite in the continuum, as is apparent from
(\ref{eq:div}) when the lattice spacing is taken to zero, $\lambda\rightarrow 0$.  
However, because this divergence is universal, 
it is  reasonable to expect that it can be consistently  removed when considering
entropy differences between two states.  
Indeed, this line of reasoning has been the basis for fundamental results such as the entropic c-theorems \cite{2004PhLB..600..142C,2012PhRvD..85l5016C}, the Bekenstein bound \cite{ 2004PhRvD..69f4006M,2008CQGra..25t5021C}, the generalized second law (GSL) of thermodynamics \cite{2012PhRvD..85j4049W}, and the quantum null energy condition (QNEC) \cite{2016PhRvD..93f4044B,2016PhRvD..93b4017B}. 

In this Letter, we show that such entropy differences can be rigorously 
defined in algebraic QFT by employing a new,
covariant regulator for entanglement entropy in an arbitrary curved spacetime\footnote{Quantum field theory can be defined on any globally hyperbolic, curved spacetime (see, e.g., \cite{Hollands:2014eia,Witten:2021jzq} for further details).} that:
\begin{enumerate}
\item \label{prop1}Is well-defined in continuum quantum field theory;
\item \label{prop2}Is manifestly covariant with respect to the subregion geometry and causal structure; 
\item \label{prop3}Satisfies the properties of entropy differences in finite-dimensional quantum mechanics, such as unitary invariance and concavity.
\end{enumerate}

This opens the door to rigorous proofs concerning entanglement. We achieve this by utilizing a tool from the theory of operator algebras known as the 
\textit{modular crossed product} \cite{takesaki1973duality}.
{In the crossed product}, one enlarges the QFT Hilbert space, $\mathcal{H}$, to $\mathcal{H}\otimes L^{2}(\mathbb{R})$ which additionally contains square-integrable wavefunctions on the real line. The crossed-product algebra then includes the position operator on $ L^{2}(\mathbb{R})$ in such a way to lead to finite, renormalized von Neumann entropies. Recently, crossed products have naturally arisen as algebras of gravitationally dressed observables in perturbative quantum gravity \cite{2022JHEP...10..008W, Chandrasekaran:2022cip,2022arXiv220910454C, Jensen:2023yxy, 2023arXiv230915897K}, and it further was shown that the entanglement entropy in these cases is equivalent to the generalized entropy.  In the present work, we will show that the crossed product is a useful tool in regulating (non-gravitational) QFT.\footnote{Distinct, but related, ideas have recently been explored in \cite{2023arXiv230607323A, 2023arXiv230609314K}. }
{All of the key arguments are presented in the main text. For completeness, mathematical details and rigorous proofs are provided in the supplemental material.}

In standard quantum mechanics, density matrices are well-defined and the entropy difference of two density matrices $\rho_{\varphi}$ and $\rho_{\psi}$ associated to states $\varphi$ and $\psi$ is defined as 
\begin{align}
\label{eq:TIDeltaS}
    \Delta S_{\textrm{vN}}(\rho_{\varphi}, \rho_{\psi}) \defn -\tr\rho_{\varphi}\log\rho_{\varphi}+ \tr  \rho_{\psi}\log \rho_{\psi}.
\end{align}
Adding and subtracting $\tr\rho_{\varphi}\log  \rho_{\psi}$ leads to
\begin{align}
\begin{aligned}
\label{eq:type1sdiff}
    \Delta S_{\textrm{vN}}(\rho_{\varphi}, \rho_{\psi}) =\Delta \braket{K_{\psi}} - S_{\textrm{rel}}(\rho_{\varphi}\parallel\rho_{\psi}) .
\end{aligned}
\end{align}
The first term is the difference in the expectation value of the (one-sided) modular Hamiltonian $K_{\psi}\defn-\log \rho_{\psi}$. The second term is the relative entropy  $S_{\textrm{rel}}
(\rho_{\varphi}\parallel\rho_{\psi})\defn \tr \rho_{\varphi} \log \rho_{\varphi} - \tr \rho_{\varphi} \log \rho_{\psi}$. Eq.~\eqref{eq:type1sdiff} is a key relation for many important results in quantum field theory. However, density matrices in local quantum field theory are undefined and so one does not have an a priori independent definition of (\ref{eq:TIDeltaS}). 

We now describe our definition of the entanglement entropy difference in quantum field theory. Given two global states $\ket{\varphi},\ket{\psi} \in \mathcal{H}$ 
we define a particular one-parameter family of states $\ket{\tilde{\varphi}_{\epsilon}}$ and $\ket{\tilde{\psi}_{\epsilon}}$ in $\mathcal{H}\otimes L^{2}(\mathbb{R})$. For any non-zero $\epsilon$, the density matrices for the states restricted to $\mathcal{R}$, $\rho_{\tilde{\varphi}_{\epsilon}}$ and $\rho_{\tilde{\psi}_{\epsilon}}$, are well-defined on the crossed product algebra. While the entanglement entropy of the individual density matrices still diverges as $\epsilon \to 0$, we show that the limit of the entanglement entropy difference is well-defined
\begin{equation}
\label{eq:DeltaS}
\Delta S_{\textrm{vN}}(\varphi,\psi)\defn -\lim_{\epsilon\to 0}\textrm{Tr}\left(\rho_{\tilde{\varphi}_{\epsilon}}\log \rho_{\tilde{\varphi}_{\epsilon}} -\rho_{\tilde{\psi}_{\epsilon}}\log \rho_{\tilde{\psi}_{\epsilon}} \right).
\end{equation}
As we describe below, the (uppercase) Tr used above is not the familiar Hilbert space trace, but a trace  --- unique up a multiplicative, state-independent constant 
on the crossed product algebra. The entropy difference \eqref{eq:DeltaS} is independent of this constant ambiguity.

Given the quantum field state $\ket{\psi} \in \mathcal{H}$, the one-parameter family of states $\ket{\tilde{\psi}_{\epsilon}}=\ket{\psi} \otimes \ket{f_{\psi,\epsilon}}$ on 
 $\mathcal{H}\otimes L^{2}(\mathbb{R})$ requires a choice of wavefunction $f_{\psi,\epsilon}(X)$ on $\mathbb{R}$. The UV divergences of the quantum field are regulated by the decay of $f_{\psi,\epsilon}(X)$ at large $X$. In order to ensure that (\ref{eq:DeltaS}) only depends on the choice of the quantum field state, we provide a prescription for uniquely specifying the one-parameter family of wavefunctions $f_{\psi,\epsilon}(X)$ given the quantum state $\psi$. Explicitly, we choose a one-parameter family of functions
 \begin{equation}
 \label{eq:f(X)}
f_{\psi,\epsilon}(X) =\sqrt{\frac{\epsilon}{2}}\sech(\epsilon(X -X_{\psi}))
\end{equation}
where $1/\epsilon$ governs the width of the function, the mean is given by
\begin{equation}
\label{eq:mean}
X_{\psi}\defn \braket{\psi|h_{\omega}^{\mathcal{A}}|\psi}
\end{equation}
and $h_{\omega}^{\mathcal{A}}$ is the quantum field theory analog of the one-sided modular Hamiltonian (defined in \eqref{eq:hmod}
{as a Hermitian form}). As we will show, this requirement for the mean of $f_{\psi,\epsilon}$ is necessary to ensure that the entropy difference \eqref{eq:DeltaS} is unitarily invariant. Thus, for each $\epsilon$, $f_{\psi,\epsilon}(X)$ is fully determined by $\psi$. The choice of a sequence $f_{\psi,\epsilon}(X)$ corresponds to a regularization scheme in the definition of (\ref{eq:DeltaS}). While we expect that a large class of such functions can be chosen to define (\ref{eq:DeltaS}), the exponential decay of (\ref{eq:f(X)}) at large $X$ yields a straightforward and rigorous proof that the limit in
(\ref{eq:DeltaS}) exists. 

Eq.~\eqref{eq:DeltaS} manifestly satisfies conditions \ref{prop1} and \ref{prop2} for entropy differences, and we prove {below} that it also satisfies condition \ref{prop3}. Furthermore, we prove that \eqref{eq:DeltaS} satisfies the following formula 
\begin{align}
\label{eq:main}
  \Delta S_{\textrm{vN}}
  (\varphi,\psi) &= \bra{\varphi}h_{\omega}^{\mathcal{A}}\ket{\varphi} -\bra{\psi}h_{\omega}^{\mathcal{A}}\ket{\psi}
      \nonumber \\
      &-S_{\textrm{rel}}\left(\varphi\parallel \omega\right)+S_{\textrm{rel}}(\psi\parallel\omega), 
\end{align}
the analog of \eqref{eq:type1sdiff} where {$\omega$ is a reference  state} (often chosen to be a global vacuum state) and
$S_{\textrm{rel}}$ is the relative entropy (defined in \eqref{eq:defsrel}). {We emphasize that both the left and right-hand sides of \eqref{eq:main} are {\em independently} well-defined within our framework.}

{We will first discuss the applications to the Bekenstein bound and the QNEC. Afterwards, we provide further details on the construction of the regulator and the proof of \eqref{eq:main}.}

\paragraph*{Proof of the Bekenstein Bound(s).---}
First, let $\mathcal{R}$ be the ``right'' Rindler wedge of Minkowski spacetime (i.e. the region $x_{1}\geq |t|$ in global inertial coordinates $(t,x_{1},\dots,x_{d-1})$). Let $\omega$ be the Minkowski vacuum; in this case, $h^{\mathcal{A}}_{\omega}$ is the usual boost operator restricted to the Rindler wedge \cite{bisognano1975duality}.
Combining \eqref{eq:main} and the fact that the relative entropy is positive semi-definite, we conclude that the vacuum subtracted entanglement entropy in a general QFT in $d$-dimensional Minkowski space is bounded as
\begin{align}
  \label{eq:bek}
  \Delta S_{\textrm{vN}}
  (\varphi,\omega) \leq 2\pi \int_{\Sigma_{\mathcal{R}}
  } d^{d-1}x~x_1 \bra{\varphi} T_{tt}(x)\ket{\varphi},
\end{align}
where $\Sigma_{\mathcal{R}}$ is the $t=0$ Cauchy surface for $\mathcal{R}$, $T_{\mu\nu}$ is the stress-energy tensor so that the right-hand side of \eqref{eq:bek} is the expectation value of $h_{\omega}^{\mathcal{A}}$. This is Casini's version of the Bekenstein bound \cite{2008CQGra..25t5021C,bekenstein1981universal}, a powerful result relating entropy differences to energy in quantum field theory. Our improvement is demonstrating that both sides of this expression can be independently well-defined.\footnote{A distinct approach may be found in \cite{2018JGP...130..113L}.}

More generally, any spacetime with a bifurcate Killing horizon --- as in the case of any black hole spacetime or de Sitter space ---
is divided by this horizon into four wedges \cite{Kay_1988}. Let $\mathcal{R}$ be the right wedge and consider the case in which the vacuum state $\omega$ is thermal (KMS) with respect to the horizon Killing vector $\xi^{\mu}$ and $h_{\omega}^{\mathcal{A}}$ is the generator of translations along $\xi^{\mu}$ in $\mathcal{R}$, namely
\begin{equation}
\label{eq:onesided}
h_{\omega}^{\mathcal{A}} = \frac{2\pi}{\kappa}\int_{\Sigma_{\mathcal{R}}}\sqrt{h}\,d^{d-1}x~n^{\mu}\xi^{\nu}T_{\mu \nu}(x)
\end{equation}
where $\sqrt{h}$ is the volume form on the Cauchy surface $\Sigma_{\mathcal{R}}$, $n^\mu$ is the unit normal to $\Sigma_{\textrm{R}}$, and $\kappa$ is the surface gravity of $\xi^{\mu}$. Then, using \eqref{eq:main}, we may prove analogous Bekenstein bounds
\begin{align}
  \label{eq:Dbound}
  \Delta S_{\textrm{vN}}
  (\varphi,\omega) \leq \frac{2\pi}{\kappa}\int_{\Sigma_{\mathcal{R}}} \sqrt{h}\,d^{d-1}x~n^{\mu}\xi^{\nu}\bra{\varphi}T_{\mu \nu}(x)\ket{\varphi},
\end{align}

\paragraph*{Proof of the QNEC.---}
If we consider a wedge-region $\mathcal{R}$  in Minkowski spacetime with the entangling surface now taken to lie on a null plane defined by the function $V(y)$ (see Figure \ref{fig:wiggly}), $h_{\omega}^{\mathcal{A}}$ is known to 
generate a boost on the null sheet about the entangling surface \cite{2012PhRvD..85j4049W,
2017JPhA...50J4001C}
\begin{align}
    h_{\omega}^{\mathcal{A}} = 2\pi \int_{\mathbb{R}^{d-2}}d^{d-2}y
    \int_{V(y)}^\infty dv\,(v-V(y))\,T_{vv}(v,y),
\end{align}
where $v$ is an affine null coordinate and $y$ denotes the transverse coordinates. Combining \eqref{eq:main} and the positivity of the second shape derivative of the relative entropy on null sheets (proved in \cite{2018arXiv181204683C}) we arrive at a \textit{point-wise} bound on 
 the null energy density (the QNEC) \begin{align}
\label{eq:QNEC}
    2\pi \bra{\varphi} T_{vv}(y)\ket{\varphi}
    \geq
     \frac{\delta^2 \Delta S_{\textrm{vN}}(\varphi,\omega)}{\delta^2 V(y)} .
\end{align}
Furthermore, 
the entropy of $\omega$ is independent of the null 
cut, so---under the assumption that the errors in the regulated
entropy have finite variational derivatives---the $\Delta$ may be removed
\begin{align}
2\pi \bra{\varphi} T_{vv}(y)\ket{\varphi}
\geq
    \lim_{\epsilon \rightarrow 0}\frac{\delta^2  S_{\textrm{vN}}(\rho_{\tilde{\varphi}_{\epsilon}})}{\delta^2 V(y)} ,
\end{align}
where $S_{\textrm{vN}}(\rho_{\tilde{\varphi}_{\epsilon}})\defn-\Tr (\rho_{\tilde{\varphi}_{\epsilon}}\log \rho_{\tilde{\varphi}_{\epsilon}})$. Thus, the local expected value of the stress-energy --- which can be negative --- is bounded below by the local variation in the vacuum subtracted entanglement entropy. In \cite{2018arXiv181204683C} a version of the QNEC, phrased as a second derivative of the relative entropy, was proven. Our contribution is to prove the QNEC phrased in terms of entanglement entropy, as in its original formulation \cite{2016PhRvD..93f4044B,2016PhRvD..93b4017B}.   
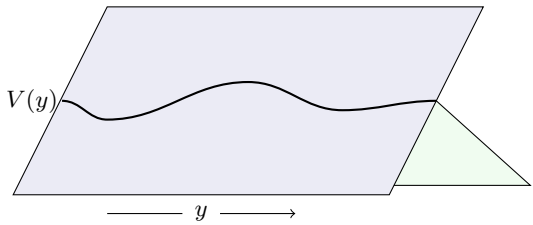
\begin{figure}
    \centering
    \begin{tikzpicture}[scale=2.5]
  \tikzset{decoration={snake,amplitude=.4mm,segment length=2mm,post length=0mm,pre length=0mm}}
\colorlet{myred}{red!80!black}
\colorlet{myblue}{blue!80!black}
\colorlet{mygreen}{green!80!black}
\colorlet{mydarkred}{red!50!black}
\colorlet{mydarkblue}{blue!50!black}
\colorlet{mylightblue}{mydarkblue!8}
\colorlet{mylightgreen}{mygreen!6}

\colorlet{myvlightblue}{mydarkblue!3}
\colorlet{mypurple}{blue!40!red!80!black}
\colorlet{mydarkpurple}{blue!40!red!50!black}
\colorlet{mylightpurple}{mydarkpurple!80!red!6}
\colorlet{myorange}{orange!40!yellow!95!black}

\fill[mylightgreen] (1/4+.03, 1/2) to[out = 0, in = 180] (2/4, 1/2-.1) to[out = 0, in = 180](5/4, 1/2+.1) to[out = 0, in = 180](7/4, 1/2-.05) to[out = 0, in = 180] (9/4,1/2) -- (9/4+1-.5,1/2+1/4-.7)-- (1/4+1-.5,1/2+1/4-.7)--cycle;
\draw[black] (1/4+.03, 1/2) to[out = 0, in = 180] (2/4, 1/2-.1) to[out = 0, in = 180](5/4, 1/2+.1) to[out = 0, in = 180](7/4, 1/2-.05) to[out = 0, in = 180](9/4,1/2) -- (9/4+1-.5,1/2+1/4-.7)-- (1/4+1-.5,1/2+1/4-.7)--cycle;

\fill[mylightblue,fill opacity = .75](0,0) -- (2,0)--(5/2,1)--(1/2,1)--cycle;
\draw[black] (0,0) -- (2,0)--(5/2,1)--(1/2,1)--cycle;
\draw[black,  thick] (1/4+.01, 1/2) to[out = 0, in = 180] (2/4, 1/2-.1) to[out = 0, in = 180](5/4, 1/2+.1) to[out = 0, in = 180](7/4, 1/2-.05) to[out = 0, in = 180](9/4,1/2);

\node[scale = 1] at (1/4-.15,1/2) {$V(y)$};
\node[scale = 1] at (1,-.1) {$y$};
\draw[->] (1.1,-.1)--(1.5,-.1);
\draw[] (.5,-.1)--(.9,-.1);

\end{tikzpicture} 
    \caption{The null plane (purple) in Minkowski spacetime is partitioned by the wiggly cut $V(y)$. This defines the wiggly Rindler wedge as 
    the region to the right of the cut, with the green plane
    forming the past null boundary.
    }
    \label{fig:wiggly}
\end{figure}
\paragraph*{Mathematical Background.---}Having described the applications 
of our proposal for regulated entropy differences,
we now describe our construction in more detail
and prove its main properties.
In order to analyze QFT from an information theoretic standpoint, we must introduce essential elements of the theory of von Neumann algebras. Given a Hilbert space, $\mathcal{H}$ and a set of operators, $\mathcal{A}$, the commutant, $\mathcal{A}'$, is defined as the set of bounded operators on $\mathcal{H}$ that commute with $\mathcal{A}$. A von Neumann algebra is a set of operators that is its own double commutant $\mathcal{A} = \mathcal{A}''$. There are three broad classes of von Neumann algebras \cite{murray1936rings}
\begin{itemize}
    \item[$\ast$]Type I von Neumann algebras are the familiar ones from quantum mechanics, containing both pure states and traces. 
    \item[$\ast$]Type II von Neumann algebras do not have pure states, though they do have traces.
    \item[$\ast$]Type III von Neumann algebras are the ones relevant for subregions in QFT and neither have pure states nor traces.
\end{itemize}
For a von Neumann algebra $\mathcal{A}$ and two (cyclic-separating\footnote{Recall that a state $\ket{\Phi} \in \mathcal{H}$ is called \emph{cyclic} if the set of states $\mathcal{A}\ket{\Phi}$ is dense in $\mathcal{H}$ and is called \emph{separating} if $a\ket{\Phi} = 0$ implies that $a = 0$ for all $a \in \mathcal{A}$.}) states $\ket{\omega}$ and $\ket{\varphi}$, the \textit{relative Tomita operator} $S_{\varphi|\omega}$ is defined as the anti-linear operator with the following action \cite{tomita1967canonical, takesaki2006tomita}
\begin{align}
\label{eq:relTom}
  S_{\varphi|\omega}\,a \ket{\omega} = a^{\dagger} \ket{\varphi}, \quad a \in \mathcal{A}.
\end{align}
The \textit{relative modular operator} $\Delta_{\varphi|\omega}$ and \textit{relative modular Hamiltonian} $h_{\varphi|\omega}$ are the following self-adjoint operators
\begin{align}
\label{eq:ModSS}
\Delta_{\varphi|\omega}=S_{\varphi|\omega}^{\dagger}\,S_{\varphi|\omega}, \quad h_{\varphi|\omega} \defn -\log \Delta_{\varphi|\omega}.
\end{align}
The relative entropy $S_{\mathrm{rel}}(\varphi\parallel\omega)$ is then defined as \cite{araki1976relative}
\begin{align}
\label{eq:defsrel}
    S_{\mathrm{rel}}(\varphi\parallel\omega) \defn \bra{\varphi}h_{\omega|\varphi}\ket{\varphi}.
\end{align}
Unlike the formula for relative entropy involving density matrices, \eqref{eq:defsrel} is well-defined even in continuum QFT (i.e. for Type III algebras).

Finally, we must introduce the crossed product algebra. Given a Type III von Neumann algebra $\mathcal{A}$ acting on a Hilbert space $\mathcal{H}$ and its modular operator $\Delta_{\omega}\defn \Delta_{\omega|\omega}$ for the state $\ket{\omega}$, we may define a new von Neumann algebra acting on the extended Hilbert space $\mathcal{H}_{\mathrm{ext.}} \defn \mathcal{H}\otimes L^2({\mathbb{R}})$ as 
\begin{align}
\label{eq:crossalg}
  \mathcal{A} \rtimes \mathbb{R} \defn \{\Delta_{\omega}^{i\hat{P}}a\Delta_{\omega}^{-i\hat{P}}, \hat{X};a \in \mathcal{A} \}'',
\end{align}
 where $\hat{X}$ and $\hat{P}$ are 
the conjugate position and momentum operators
acting on $L^2({\mathbb{R}})$.
This extended algebra is a Type II algebra \cite{takesaki1973duality} with trace, 
\begin{align}
  \label{eq:IIinfty}
  \Tr\,\hat{a} = \int_{\mathbb{R}} dX~e^X \bra{\omega,X} \hat{a}\ket{\omega,X}, \quad \hat{a} \in \mathcal{A} \rtimes \mathbb{R}.
\end{align}
where $\ket{\omega,X}\defn \ket{\omega}\otimes \ket{X}$ and $\ket{X}$ is an eigenstate of $\hat{X}$ with eigenvalue $X$. The trace (\ref{eq:IIinfty}) is defined up to an overall multiplicative, state-independent constant. The entropy difference given by \eqref{eq:TIDeltaS} will not depend on this ambiguity. That this expression satisfies the cyclicity properties of a trace (i.e., $\Tr(\hat{a}\hat{b} )= \Tr (\hat{b}\hat{a})$) is nontrivial. A proof can be found in \cite{2022JHEP...10..008W}. Remarkably, adding the auxiliary variable $\hat{X}$ to the algebra has opened the door to having well-defined traces and hence density matrices and von Neumann entropies.

\paragraph*{Entropy differences in QFT.---}

We are now ready to present our prescription for entropy differences. We consider the following classical-quantum states in  $\mathcal{H}_{\mathrm{ext.}}$,
\begin{align}
  \ket{\tilde{\varphi}_\epsilon} \defn  \int_{\mathbb{R}} dX f_{\varphi,\epsilon}(X)\ket{\varphi} \otimes \ket{X},
\end{align}
where $f_{\varphi,\epsilon}$ was defined in \eqref{eq:f(X)}.
The regulating effect of the wavefunction can be seen from examining the structure of the crossed product algebra \eqref{eq:crossalg}.
The quantum field theory operators are smeared in (modular) time, {with the width  $\epsilon$ of the Fourier transformed wavefunction $\tilde{f}_{\varphi,\epsilon}(P)$} determining the amount of smearing. As $\epsilon \rightarrow 0$, we recover local operators. Smearing local operators in time is known to be sufficient for regulating divergences \cite{borchers1964field,2023arXiv230202709S} and, indeed, the universal short-distance divergences in correlation functions become smooth for all operators in the crossed product algebra with finite trace, with the wavefunction providing finite smearing to all such operators \cite{2023arXiv230302837W,2023arXiv22XXXXXXXK2}.

Recall that we fixed the mean $X_{\varphi}$ of $f_{\varphi,\epsilon}$ to be $X_{\varphi} \defn \bra{\varphi} h_{\omega}^{\mathcal{A}}\ket{\varphi}$
where we assume that modular Hamiltonian can be decomposed as 
\begin{equation}
\label{eq:hmod}
h_{\omega} = h_{\omega}^{\mathcal{A}}+h_{\omega}^{\mathcal{A}'}
\end{equation}
with $h_{\omega}^{\mathcal{A}}$ and $h_{\omega}^{\mathcal{A}'}$ %
{\em Hermitian forms} 
on $\mathcal{H}$ affiliated with $\mathcal{A}$ and $\mathcal{A}^{\prime}$ respectively. Thus, while they are not defined as operators on $\mathcal{H}$, 
matrix elements of $h_{\omega}^{\mathcal{A}}$ and $h_{\omega}^{\mathcal{A}^{\prime}}$ are well-defined for a dense set of states and, furthermore, they satisfy $[h_{\omega}^{\mathcal{A}},a^{\prime}]=0=[h_{\omega}^{\mathcal{A}^{\prime}},a]$ for $a\in \mathcal{A}$ and $a^{\prime}\in \mathcal{A}^{\prime}$ (see section \ref{subsec:sesq} of the supplemental material).
While well-defined, this split is not unique because $h_{\omega}^{\mathcal{A}}$ can be augmented by an operator-valued distribution localized at the boundary.\footnote{See \cite{2019PhRvD..99l5020C} for relevant discussions.} This ambiguity drops out of the physical statements such as the Bekenstein bound and QNEC. There are many explicitly known examples of $h_{\omega}$ satisfying such a splitting. An important example is when $h_{\omega}$ is the boost generator and so $h_{\omega}^{\mathcal{A}}$ is explicitly of the form  \eqref{eq:onesided}. Near the entangling surface, $h_{\omega}$ is expected to locally approach the boost generator, which suggests that the 
existence of a decomposition is generic. 

The density matrices for $\ket{\tilde{\varphi}_\epsilon}$ and $\ket{\tilde{\omega}_\epsilon}$ are given by \cite{2022JHEP...10..008W, Jensen:2023yxy}
\begin{align}
  \rho_{\tilde{\varphi}_\epsilon} &= f_{\varphi,\epsilon}(\hat{X})e^{-i\hat{P}h_{\omega}}e^{-\hat{X}}\Delta_{\omega|\varphi}e^{i\hat{P}h_{\omega}}f_{\varphi,\epsilon}(\hat{X})
  \\
  \rho_{\tilde{\omega}_\epsilon} &= f_{\omega,\epsilon}(\hat{X})^2e^{-\hat{X}}.
  \label{eqn:rhovac}
\end{align}
The entropy of the regulated vacuum is
\beq
S_{\vn}(\rho_{\tilde{\omega}_\epsilon}) = X_{\varphi}- 
\int_{\mathbb{R}} dX f_{\omega,\epsilon}({X})^2\log f_{\omega,\epsilon}({X})^2.
\eeq
The divergence coming from the local QFT limit $\epsilon
\rightarrow 0$ is logarithmic, consistent with 
$\epsilon$ being a dimensionless regulator. 
The absence of more singular power-law divergences is 
 related to the fact that crossed-product entropies
are most naturally interpreted as entropy differences, and hence
state-independent area-law divergences should not appear. 

The entropy of the regulated excited state is more
involved due to the need to take a logarithm
of noncommuting operators appearing in the 
density matrix $\rho_{\tilde{\varphi}_\epsilon}$.  
The exact expression for this entropy is derived 
in upcoming work by one of us 
\cite{Faulkner2023}, and for completeness
the derivation is reproduced in section \ref{app:bound} of the supplemental
material. The result is
\begin{multline}
    S_{\vn}(\rho_{\tilde\varphi_\epsilon})
= 
S_{\vn}(\rho_{\tilde\omega_\epsilon})-
S_\text{rel}(\varphi\parallel\omega) 
+\langle{\hat{X}}\rangle_{\tilde\varphi_\epsilon}
 \\
+\left\langle2(\epsilon h_{\omega|\varphi} \coth(\epsilon h_{\omega|\varphi})-1) - 
R_\epsilon(\Delta_\varphi, h_{\omega|\varphi})
\right\rangle_\varphi
\label{eqn:Sexact}
\end{multline}
where we have used the wavefunction $f_{\varphi,\epsilon}(\hat{X})$ given in \eqref{eq:f(X)} and  $R_\epsilon(\Delta_\varphi, h_{\omega|\varphi})$
is a positive operator. Since the map that sends
$|\varphi\rangle$ to the  state 
$|\tilde{\varphi}_f\rangle = |\varphi\rangle \otimes
|f\rangle$ for fixed wavefunction $f(X)$
is a quantum
channel from states on the QFT algebra to states 
on the crossed product algebra, monotonicity of 
relative entropy for this channel implies that the expression
on the second line of (\ref{eqn:Sexact}) must be positive.
Together with the positivity of $R_\epsilon(\Delta_\varphi,h_{\omega|\varphi})$, 
this shows that the second line of (\ref{eqn:Sexact})
is bounded between $0$ and 
$
\langle 2(\epsilon h_{\omega|\varphi} \coth(\epsilon h_{\omega|\varphi})-1)\rangle_\varphi$, which itself is bounded by $2\epsilon(S_{\textrm{rel}}(\varphi\parallel\omega)+2)$.  Hence,
for states with finite relative entropy, this term vanishes
in the local QFT limit. Recalling (\ref{eq:mean}),
we arrive at the desired expression for continuum entropy
differences,
\beq
\label{eq:vac_subtract}
\Delta S_{\vn}(\varphi,\omega) = 
\langle\varphi|h_\omega^{\mathcal{A}}|\varphi\rangle
-S_{\textrm{rel}}(\varphi\parallel\omega)
\eeq
Let $\psi$ be another state with $\langle\psi|h_\omega^{\mathcal{A}}|\psi\rangle$ and $S_{\textrm{rel}}(\psi\parallel\omega)$ finite. We then define the entropy difference $\Delta S_{\textrm{vN}}(\varphi,\psi)$ as
\begin{equation}
    \label{eq:postmain}
    {\Delta S_{\textrm{vN}}(\varphi, \psi) = \Delta S_{\textrm{vN}}(\varphi, \omega) - \Delta S_{\textrm{vN}}(\psi, \omega)} ,
\end{equation}
which immediately gives \eqref{eq:main}. In summary, the local QFT limit $\epsilon \to 0$  of differences of von Neumann entropies of density matrices on Type II algebras reproduces the expected result for differences in von Neumann entropies in Type III algebras. 

\paragraph*{Properties.---}There are various properties that entropy differences must have in order to be physically meaningful.
To begin, it is manifest from the definitions \eqref{eq:DeltaS}  that the entropy difference is both antisymmetric
\begin{align}
  \Delta S_{\textrm{vN}}(\varphi_1, \varphi_2) = -\Delta S_{\textrm{vN}}(\varphi_2, \varphi_1) 
\end{align}
and transitive
\begin{align}  
  \Delta S_{\textrm{vN}}(\varphi_1, \varphi_3) = \Delta S_{\textrm{vN}}(\varphi_1, \varphi_2) +\Delta S_{\textrm{vN}}(\varphi_2, \varphi_3).
\end{align}
The antisymmetry implies that the entropy difference of a state with itself is zero. Notably, unlike the relative entropy, the entropy difference being zero does not imply that the states are equivalent. Indeed, a crucial property of entanglement entropy is that it is invariant under local unitary transformations of the state. Such a transformation only mixes the degrees of freedom in the algebra within itself, so there is no information gained or lost. 

For unitary operators $U \in \mathcal{A}$ and $U' \in \mathcal{A}'$, it is straightforward to check that the following relations for the relative Tomita
operators hold
\begin{align}
\begin{aligned}
S_{UU'{\varphi}|{\omega}} = U'\,S_{\varphi|\omega}\,U^{\dagger},\quad S_{{\varphi}|UU'{\omega}} = U\,S_{\varphi|\omega}\,U'^{\dagger},
\end{aligned}
\end{align}
and thus, by \eqref{eq:ModSS}, 
\begin{align}
\Delta_{{\varphi}|UU'{\omega}} = U'\,\Delta_{\varphi|\omega}\,U'^{\dagger},\quad \Delta_{{UU'\varphi}|{\omega}} =U\,\Delta_{\varphi|\omega}\,U^{\dagger}.
\end{align}
The difference in relative entropies is then
\begin{align}
\begin{aligned}
  S_{\textrm{rel}}(\varphi\parallel\omega) &-S_{\textrm{rel}}(UU'\varphi\parallel\omega) 
  = \bra{\varphi}(h_{\omega|\varphi}- h_{{U^{\dagger}}{\omega}|\varphi} )\ket{\varphi}
   \\
  &=\bra{\varphi}(( h_{\omega} - h_{{U^{\dagger}}{\omega}} )- 
  (h_{\varphi|\omega}
  - 
    h_{\varphi|{U^{\dagger}}{\omega}}))\ket{\varphi}
    \\
    &= 
      \bra{\varphi}( h_{\omega} - U^{\dagger}h_{{\omega}}U )\ket{\varphi},
\end{aligned}
\end{align}
where in the second line, we have again used the identity $h_{\varphi|\omega} -h_{\omega} = h_{\varphi}-h_{\omega|\varphi}$.
Due to the decomposition $h_{\omega} = h_{\omega}^{\mathcal{A}}+h_{\omega}^{\mathcal{A}'} $, we have
\begin{equation}
       S_{\textrm{rel}}(\varphi\parallel\omega) -S_{\textrm{rel}}(U\varphi\parallel\omega)
=\bra{\varphi}h_{\omega}^{\mathcal{A}}\ket{\varphi} -\bra{\varphi}U^{\dagger}h_{\omega}^{\mathcal{A}}U\ket{\varphi}. 
\end{equation}
From \eqref{eq:main}, we conclude that the entropy difference is unitarily invariant. Moreover, it would not have been unitarily invariant had we chosen any other value for the mean $X_{\varphi}$.

We next prove that $\Delta S_{\textrm{vN}}$ cannot decrease from classically mixing states. This follows directly from the fact that the relative entropy is jointly convex in its arguments \cite{araki1976relative}, so using \eqref{eq:vac_subtract} one finds
  \begin{align}
    \Delta S_{\textrm{vN}}(\sum_i \lambda_i \varphi_i,\omega) &\geq \sum_i\lambda_i\left[ \bra{\varphi_i}h_{\omega}^{\mathcal{A}} \ket{\varphi_i}-S_{\textrm{rel}}\left(\varphi_i,\omega\right)\right]
    \nonumber
    \\  
    &=\sum_i \lambda_i \Delta S_{\textrm{vN}}(\varphi_i,\omega) .
  \end{align}

\paragraph*{Discussion.---}
\label{sec:disc}
In this Letter, we have demonstrated that entropy differences in quantum field theory are mathematically well-defined and obey several desirable properties. This led to rigorous proofs of the Bekenstein bound and the QNEC. There are two exciting avenues for future directions. The first is to rigorously prove other theorems that have been previously derived using heuristic entropy differences, such as c-theorems and the generalized second law.\footnote{The case of the generalized second law is addressed in \cite{Faulkner2023}.} This will require addressing challenges such as comparing the entropies of the same state but for different spacetime regions. The other direction is that our formalism can be leveraged to prove new properties about quantum field theory and quantum gravity. 

\acknowledgments
We thank Tom Faulkner, Ted Jacobson, 
Nima Lashkari, Ronak Soni, 
Jonathan Sorce, and Edward Witten for discussions. 
JKF is supported by the Marvin L.~Goldberger Member Fund at the Institute for Advanced Study and the National Science Foundation under Grant No. PHY-2207584. S.L.~acknowledges the support of NSF grant No. PHY-2209997. A.A.R.~is supported by the Stanford Institute of Theoretical Physics and the National Science Foundation under Grant No. PHY-1720397. G.S.~and S.L.~are supported by the Princeton Gravity Initiative at Princeton University. A.J.S.~is supported by the Air Force
Office of Scientific Research under award number FA9550-19-1-036.
A.J.S. would like to thank the Isaac Newton Institute for Mathematical Sciences for support and hospitality during the 
programme ``Black holes:\ bridges between number theory and holographic quantum information''
when work on this paper was undertaken.  This work 
was supported by EPSRC Grant Number EP/R014604/1.

\bibliography{main}
\onecolumngrid
\section*{Supplemental Material}
\beginsupplement

\section{Comments on Hermitian forms and domains }
\label{subsec:sesq}
As noted in the main text, the integral of the stress-energy flux over a partial Cauchy surface --- as in the right-hand side of (\ref{eq:Dbound}) --- has infinite fluctuations and therefore is not a well-defined, unbounded operator. Similarly, the one-sided modular Hamiltonian $h_{\omega}^{\mathcal{A}}$ cannot be defined as an operator on $\mathcal{H}$. Nevertheless, their matrix elements can be defined on a class of states and so they can be defined as {\em sesquilinear forms} on the Hilbert space $\mathcal{H}$. In this subsection, we will review the basic definitions and properties of sesquilinear forms as well as our assumptions regarding the domain and properties of the one-sided modular Hamiltonian. 

A {\it sesquilinear form} is a map
$\mathfrak{s}:\mathcal{D}(\mathfrak{s})\times \mathcal{D}(\mathfrak{s})
\rightarrow \mathbb{C}$ that is conjugate-linear
in the first argument and linear in the second,
where the subspace 
$\mathcal{D}(\mathfrak{s})\subseteq
\hs$ is known as the {\it domain of $\mathfrak{s}$} \cite{Schmudgen2012}.  If the form satisfies 
$\mathfrak{s}(|\phi\rangle,|\chi\rangle) = 
\overline{\mathfrak{s}(|\chi\rangle, |\phi\rangle)}$,
it is called a {\it Hermitian form}.  
Given any unbounded self-adjoint 
operator
$A$ with domain $\mathcal{D}(A)$, one can define
its associated Hermitian form $\mathfrak{s}_A$
by the relation $\mathfrak{s}_A(|\phi\rangle,|\chi\rangle)
=\langle \phi|A|\chi\rangle$ with domain
$\mathcal{D}(\mathfrak{s}_A) = \mathcal{D}(|A|^{\frac12})
\supsetneq \mathcal{D}(A) $.  Notably, the domain
of the form $\mathfrak{s}_A$ is strictly larger
than the domain of the operator $A$ when it
is unbounded.

Since the notion of a sesquilinear form may be unfamiliar to a general audience, we now give some representative examples. In standard quantum mechanics a familiar example 
 is the operator-valued distribution $\delta(\hat{X}-X_{0})$ on $L^{2}(\mathbb{R})$ which is defined as a Hermitian form $\mathfrak{s}_{X_{0}}$ that acts on wavefunctions $f,g\in L^{2}(\mathbb{R})$ by $\mathfrak{s}_{X_{0}}(f,g)\defn \overline{f(X_0)} g(X_0)$. The domain $\mathcal{D}(\mathfrak{s}_{X_{0}})$ is the space of square-integrable wavefunctions on $\mathbb{R}$ that are continuous at $X_{0}$. This space is dense in $L^{2}(\mathbb{R})$. Therefore, while $\delta(\hat{X}-X_{0})$ cannot be defined as an unbounded operator it can be defined as a sesquilinear form. 
 
Similarly, in quantum field theory, the local operator-valued
 distribution $\op(x)$ cannot be defined as an unbounded operator, but it can always be defined as a sesquilinear form on a dense domain. Indeed, one can expect there to be a large class of states on which the matrix elements $\braket{\psi|\mathcal{O}(x)|\chi}$ are finite since the vacuum correlation functions of $\mathcal{O}(x)$ at spacelike separated points are finite. Thus we could define a domain of $\mathcal{O}(x)$ to consist of states obtaining be acting with bounded functions of operators smeared in an open spacetime regions separated from $x$ by a finite separation $\varepsilon$. By the Reeh-Schlieder theorem \cite{Reeh:1961ujh}, such a domain is dense in $\mathcal{H}$. One can straightforwardly extend $\mathcal{O}(x)$ to a larger domain by simply decreasing the value of $\epsilon$ in the previous construction. Following standard terminology, we will refer to $\mathcal{O}(x)$ as a {\em local operator}.  While $\mathcal{O}(x)$ is not a true operator on the Hilbert space, the form $\mathcal{O}(x)$ enjoys many properties of the quantum field operators. Indeed, if $\mathcal{A}$ is the algebra of quantum fields in a region $\mathcal{R}$ and $\mathcal{A}^{\prime}$ is the commutant of $\mathcal{A}$, then if $\mathcal{O}(x)$ lies in $\mathcal{R}$ we have that 
 \begin{equation}
 \label{eq:commsesq}
 [a^{\prime},\mathcal{O}(x)] = 0
 \end{equation}
where (\ref{eq:commsesq}) is defined in the sense of matrix elements on the domain of $\mathcal{O}(x)$ and $a^{\prime}\in \mathcal{A}^{\prime}$ are elements of the commutant that preserve this domain. 

 Since local operators can be defined as sesquilinear forms, the integral of local operators can also be defined as a sesquilinear form where the domain of the integrated operator can be defined as the intersection of the domain of the local operators appearing in the integral over the range of integration. An example relevant to considerations of this paper is the case where the spacetime contains a bifurcate Killing horizon with Killing vector $\xi^{\mu}$ and the local operator is $T_{\mu \nu}\xi^{\nu}$. In this case, the region $\mathcal{R}$ is the right wedge and one can define
 \begin{equation}
 \label{eq:hxi}
h_{\xi} = \frac{2\pi}{\kappa}\int_{\Sigma}\sqrt{h}\,d^{d-1}x~ n^{\mu}\xi^{\nu}T_{\mu\nu}
 \end{equation}
 where $\Sigma$ is a complete Cauchy surface for the spacetime and $\kappa$ is the surface gravity of $\xi^{\mu}$. As is well-known, $h_{\xi}$ is a true (unbounded) operator defined on a dense domain of the Hilbert space where $h_{\xi}$ acts by boosts around the bifurcation surface. One can split (\ref{eq:hxi}) into the sum of two contributions 
 \begin{equation}
 \label{eq:hxiLR}
 h_{\xi} = h_{\xi}^{\mathcal{R}^{\prime}} + h_{\xi}^{\mathcal{R}}
 \end{equation}
 where 
 \begin{equation}
 \label{eq:hxiR}
 h_{\xi}^{\mathcal{R}} = \frac{2\pi}{\kappa}\int_{\Sigma_{\mathcal{R}}}\sqrt{h}\,d^{d-1}x~ n^{\mu}\xi^{\nu}T_{\mu\nu}
 \end{equation}
 where $\Sigma_{\mathcal{R}}$ is a spacelike Cauchy surface for $\mathcal{R}$ and $h_{\xi}^{\mathcal{R}^{\prime}}$ is defined similarly with  $\Sigma_{\mathcal{R}}$ replaced by a Cauchy surface $\Sigma_{\mathcal{R}^{\prime}}$ for the causal complement $\mathcal{R}^{\prime}$ of $\mathcal{R}$ (also known as the left wedge in this case). 
 The one-sided integral does not define a true operator due to ultraviolet correlations across the bifurcation surface. However $h_{\xi}^{\mathcal{R}}$ --- and similarly $h_{\xi}^{\mathcal{R}^{\prime}}$ --- are well-defined as sequilinear forms on a dense domain of states. Furthermore, since $h_{\xi}^{\mathcal{R}}$ is an integral of a local form, it satisfies 
 \begin{equation}
 [a^{\prime},h_{\xi}^{\mathcal{R}}]=0
 \end{equation}
in the sense of matrix elements for any $a^{\prime}\in \mathcal{A}^{\prime}$ that preserves the domain of $h_{\xi}^{\mathcal{R}}$. A similar statement holds for  $h_{\xi}^{\mathcal{R}^{\prime}}$ which commutes with elements $a\in \mathcal{A}$. Therefore, we say that $h_{\xi}^{\mathcal{R}}$ is a Hermition form {\em  affiliated} with $\mathcal{A}$  and $h_{\xi}^{\mathcal{R}^{\prime}}$ is affiliated with $\mathcal{A}^{\prime}$.

If the state $\ket{\omega}$ is thermal (KMS) in $\mathcal{R}$ and its causal complement $\mathcal{R}^{\prime}$ then the modular Hamiltonian $h_{\omega}$ (defined in (\ref{eq:ModSS})) is equivalent to $h_{\xi}$ and (\ref{eq:hxiLR}) corresponds to a splitting of the modular operator into two forms associated to $\mathcal{R}$ and $\mathcal{R}^{\prime}$ respectively. More generally, we assume in this letter that, for {\em any} state $\omega$ in a general subregion $\mathcal{R}$, the modular Hamiltonian $h_{\omega}$ admits such a splitting into one-sided modular Hamiltonians
\begin{equation}
\label{eq:modHamsplit}
h_{\omega} = h_{\omega}^{\mathcal{A}^{\prime}} + h_{\omega}^{\mathcal{A}}
\end{equation}
where $h_{\omega}^{\mathcal{A}}$ is a Hermitian form affiliated with $\mathcal{A}$ and $h_{\omega}^{\mathcal{A^{\prime}}}$ is a Hermitian form affiliated with $\mathcal{A}^{\prime}$ on dense domains of the Hilbert space. Thus, we assume that they satisfy the property\footnote{The splitting of the modular Hamiltonian into the sum of Hermitian forms was required when proving unitary invariance of the entropy difference \ref{eq:DeltaS}. In this proof we need only the weaker condition that, for any unitaries $U\in \mathcal{A}$ and $U^{\prime}\in \mathcal{A}^{\prime}$, $\braket{\varphi|U^{\dagger}[h_{\omega}^{\mathcal{A}^{\prime}},U]|\varphi}=0=\braket{\varphi|U^{\prime\dagger}[h_{\omega}^{\mathcal{A}},U^{\prime}]|\varphi}$.}
\begin{equation}
\label{eq:ahcomm}
[a^{\prime},h_{\omega}^{\mathcal{A}}] = 0 \quad \textrm{ and} \quad [a,h_{\omega}^{\mathcal{A}^{\prime}}] = 0
\end{equation}
in the sense of matrix elements for any $a\in\mathcal{A}$ and $a^{\prime}\in\mathcal{A}^{\prime}$. These properties reflect the general expectation that the one-sided modular Hamiltonians for generic regions behave similarly to the one-sided boost generators (\ref{eq:hxiR}) near the entangling surface.

\section{Details on regulated entropy}
\label{app:bound}

In this appendix we fill in the details demonstrating
that entropy differences defined using
a crossed product regulator given by \eqref{eq:DeltaS} are well-defined and satisfy the crucial relation
(\ref{eq:main}) in the limit that the regulator
is taken to zero.  In order to show this, we use
the exact expressions derived in \cite{Faulkner2023}
for the entropy in the type $\ttwo$
crossed product algebra for states of the form
$|\tilde{\varphi}_f\rangle = 
|\varphi\rangle\otimes|f\rangle$, where $|f\rangle$ is the 
state associated with a positive, nowhere-vanishing,
normalized
wavefuntion $f(X)$ in $L^2(\mathbb{R})$.  This 
allows us to derive bounds on the 
corrections to the leading order expression for the entropy,
and when using the wavefunction $f_{\varphi,\epsilon}$ given by  \eqref{eq:f(X)}, this bound is determined by 
the relative entropy $S_\text{rel}(\varphi\parallel \omega)$
with respect to the vacuum state $|\omega\rangle$, and 
is suppressed by $\epsilon$.  Hence, for states with 
finite relative entropy, the limiting expression for 
the entropy exists and agrees with the semiclassical 
expression found in \cite{Chandrasekaran:2022cip,
2022arXiv220910454C, Jensen:2023yxy}.

The exact density matrix for the state $|\tilde{\varphi}_f\rangle$
on the crossed product algebra was obtained in 
\cite{Jensen:2023yxy}, and is given by
\beq
\rho_{\tilde{\varphi}_f} = f(\hat{X}) e^{-\frac{\hat{X}}{2}}
e^{-i\hat{P} h_\omega}
\Delta_\omega^{-\frac12}
\Delta_{\varphi|\omega}
\Delta_\omega^{-\frac12}
e^{i\hat{P}h_\omega}
e^{-\frac{\hat{X}}{2}} f^*(\hat{X}).
\eeq
The modular operators appearing in this expression
satisfy the cocycle relation 
$\Delta_{\omega}^{-\frac12}
\Delta_{\varphi|\omega}^{\frac12}=
\Delta_{\omega|\varphi}^{-\frac12}\Delta_{\varphi}^{\frac12}$,
and both expressions are operators affiliated with 
$\mathcal{A}$ (see, e.g., \cite[Appendix C]{Jensen:2023yxy}). 
Because the flows generated by $h_\omega$ and $h_{\omega|\varphi}$
agree on operators affiliated with $\mathcal{A}$, 
we can replace $e^{-i\hat{P}h_\omega}$ with 
$e^{-i\hat{P}h_{\omega|\varphi}}$ in the above expression.  
This, along with the cocycle relation, 
then leads to the equivalent expression for the density 
matrix
\beq \label{eqn:rhoequiv}
\rho_{\tilde{\varphi}_f}
=
e^{-i\hat{P}h_{\omega|\varphi}}
f(\hat{X}+h_{\omega|\varphi})e^{-\hat{X}}\Delta_\varphi
f^*(\hat{X}+h_{\omega|\varphi})e^{i\hat{P}h_{\omega|\varphi}}
\eeq
For brevity, we now define 
\beq
h \defn h_{\omega|\varphi}, \quad 
\Delta \defn
\Delta_\varphi.
\eeq
The logarithm of this density matrix can be written
\begin{align}
    - \log \rho_{\tilde{\varphi}_f} 
    &= -e^{-i\hat{P}h} \log \left[ f(\hat{X}+h) e^{-\hat{X}}\Delta f^*(\hat{X}+h)\right]e^{i\hat{P}h} 
    \nonumber
    \\
    &= - e^{-i\hat{P}h} \left( -\hat{X}+\log \left[ f(\hat{X}+h) \Delta f^*(\hat{X}+h)\right]\right)e^{i\hat{P}h}
    \nonumber
    \\
    &= \hat{X}-h- \log \left[ f(\hat{X}) e^{-i\hat{P}h}
    \Delta e^{i\hat{P}h}f^*(\hat{X})\right]
    \nonumber
    \\
    &= \hat{X}-h - \log (g(\hat{X})\Delta )
    - \log W+ \log (g(\hat{X})\Delta ) 
    \label{eqn:logrho}
\end{align}
where we have defined
\beq
W \defn f(\hat{X}) e^{-i\hat{P}h}\Delta e^{i\hat{P}h} f^*(\hat{X}),
\qquad
g(\hat{X}) \defn |f(\hat{X})|^2.
\eeq

The terms $- \log W+ \log (g(\hat{X})\Delta )$ are 
responsible for the 
subleading corrections to the 
entropy as $\epsilon\rightarrow 0$ 
that we seek to control. As we will eventually compute the expectation
value of these terms in the state $|\varphi\rangle
\otimes |f\rangle$, we can make progress
on evaluating them by passing them through
the unital completely positive (UCP) map $\langle f|\cdot|f\rangle$,
i.e.\ taking the partial expectation value in the 
state $|f\rangle$. The result of this procedure is the 
following lemma:

\begin{lemma}[\cite{Faulkner2023}] \label{lem:Rdelta} The image
of the operator $-\log W+\log(g(\hat{X})\Delta)$
under the UCP map $\langle 
f|\cdot|f\rangle$ is given by 
\beq\label{eqn:entcorr}
\langle f|-\log W + \log(g(\hat{X})\Delta)|f\rangle
=
-R(\Delta,h) - \int_{-\infty}^\infty dz\, \delta(z,h)\log g(z)
\eeq
where 
\beq
\delta(z,h)\defn g(z+h) - g(z)
\eeq
and $R(\Delta,h)$
is a positive operator defined by
\beq
R(\Delta,h) \defn \int_{-\infty}^\infty dz\, g(z)
\int_0^\infty d\mu
\,\frac{1}{\mu+\Delta} \frac{\delta(z,h)}{g(z)} 
\left(
\frac{\mu}{\frac{\mu}{\Delta} + \frac{g(z+h)}{g(z)}} 
\right)\frac{\delta(z,h)}{g(z)} \frac{1}{\mu+\Delta}.
\eeq
\end{lemma}

\begin{proof}
We will evaluate the logarithms using the integral expression
\begin{align}
-\log A +\log B= \int_0^\infty d\lambda \left(\frac{1}{\lambda +A}
- \frac{1}{\lambda +B} \right).
\end{align}
Noting that 
\begin{align}
\frac{1}{\lambda + W} 
= 
\frac{1}{f^*(\hat{X})} \frac{1}{\frac{\lambda}{g(\hat{X})}
+ e^{-i\hat{P}h} \Delta e^{i\hat{P}h}}
\frac{1}{f(\hat{X})}
=
\frac{1}{f^*(\hat{X})} e^{-i\hat{P}h}
\frac{1}{\frac{\lambda}{g(\hat{X}+h)} + \Delta}
e^{i\hat{P}h} \frac{1}{f(\hat{X})},
\end{align}
we find
\begin{align}
\begin{aligned}
-\log W + \log (g(\hat{X})\Delta)
&=
\int_0^\infty d\lambda\left(\frac{1}{\lambda + W}-
\frac{1}{\lambda + g(\hat{X}) \Delta}\right)
\\
&=
\int_0^\infty d\lambda\,\frac{1}{f^*(\hat{X})}
\left(e^{-i\hat{P}h} \frac{1}{\frac{\lambda}{g(\hat{X}+h)}
+ \Delta} e^{-i\hat{P}h}
- \frac{1}{\frac{\lambda}{g(\hat{X})} + \Delta}
\right) \frac{1}{f(\hat{X})}.
\end{aligned}
\end{align}
Next, we pass the operator through the UCP map
$\langle f|\cdot |f\rangle$. Using
\begin{align}
\frac{1}{f(\hat{X})}|f\rangle =\sqrt{2\pi}|0_P\rangle,
\end{align}
where $\ket{0_P}$ is the (improper) eigenstate of $\hat{P}$ with zero eigenvalue,
and recalling the operator identity
\begin{align}
A^{-1} + B^{-1} = A^{-1} (A+B)B^{-1},
\end{align}
we obtain
\begin{align}\begin{aligned}
\langle f|&-\log W + \log (g(\hat{X})\Delta)|f\rangle
=
2\pi \int_0^\infty d\lambda \,\langle 0_P|
\left(
e^{-i\hat{P}h}\frac{1}{\frac{\lambda}{g(\hat{X}+h)}
+\Delta} e^{i\hat{P}h} 
-\frac{1}{\frac{\lambda}{g(\hat{X})}+\Delta}
\right)|0_P\rangle
\\
&=
\int_0^\infty d\lambda \int dz
\left(\frac{1}{\frac{\lambda}{g(z+h)} +\Delta}
-\frac{1}{\frac{\lambda}{g(z)}+\Delta}
\right)
\\
&=
\int_0^\infty d\lambda\int dz\left(
\frac{1}{\frac{\lambda}{g(z)} +\Delta}
\left(\frac{\lambda}{g(z)}-\frac{\lambda}{g(z+h)}\right)
\frac{1}{\frac{\lambda}{g(z+h)} +\Delta}
\right)
\\
&=
\int_0^\infty d\lambda\int dz\left(
\frac{1}{\lambda + \Delta g(z)}
\delta(z,h)
\frac{\lambda}{\lambda +\Delta g(z+h)}
\right)
\\
&=
\int_0^\infty d\lambda\int dz\left(
\frac{-1}{\lambda +\Delta g(z)} 
\delta (z,h)\Delta g(z+h) 
\frac{1}{\lambda + \Delta g(z+h)} 
+\frac{1}{\lambda + \Delta g(z)}\delta(z,h)
\right) 
\label{eqn:convergent2}
\end{aligned}
\end{align}
where the integration over $z$ is over the entire real line. It is important that the $\lambda$ integral is 
on the outside, since it does not converge at 
fixed $z$; it only converges after integrating
the expression over $z$. The last line then 
separates off a piece where the $\lambda$ 
integral does converge at fixed $z$, plus an
additional term in which the $\lambda$ integral
can be performed explicitly. Because $\delta(z,h)$ averages to zero
\begin{align}
\int dz\,\delta(z,h) = \int dz\,(g(z+h)-g(z)) = 0,
\end{align}
we can subtract $\frac{1}{\lambda + \Delta} 
\delta(z,h)$ from the leftover term and perform
the $\lambda$ integral,
\begin{align}
\begin{aligned}
\label{eqn:leftover}
\int_0^\infty d\lambda \int dz
\frac{1}{\lambda +\Delta g(z)} \delta(z,h)
&=
\int_0^\infty d\lambda \int dz 
\left(\frac{1}{\lambda +\Delta g(z)} - 
\frac{1}{\lambda
+\Delta}\right)\delta(z,h)
\\
&= -\int dz\,\delta(z,h) \log g(z).
\end{aligned}
\end{align}
For the first term in (\ref{eqn:convergent}), we can 
switch the order of integration
since the $\lambda$ integral is now convergent
at fixed $z$, which allows 
the substitution $\mu = \frac{\lambda}{g(z)}$ 
\begin{align}
\begin{aligned}
\label{eqn:convergent}
\int_0^\infty d\lambda \int dz &
\,\frac{-1}{\lambda +\Delta g(z)} 
\delta (z,h)\Delta g(z+h) 
\frac{1}{\lambda + \Delta g(z+h)}
\\
&=
-\int dz \int_0^\infty d\mu \frac{1}{\mu+\Delta}
\delta(z,h) \Delta \frac{1}{\frac{\mu g(z)}{g(z+h)}
+ \Delta}
\\
&=
-\int dz \int_0^\infty d\mu
\frac{1}{\mu+\Delta}
\delta(z,h) \Delta 
\left(
\frac{1}{\frac{\mu g(z)}{g(z+h)}
+ \Delta}
-\frac{1}{\mu+\Delta}\right)
\\
&=
-\int dz \int_0^\infty d\mu
\frac{1}{\mu+\Delta}
\delta(z,h) \Delta 
\frac{1}{\frac{\mu g(z)}{g(z+h)} +\Delta}
\mu\left(1-\frac{g(z)}{g(z+h)}\right) \frac{1}{\mu+\Delta}
\\
&=
-\int dz\,g(z) \int_0^\infty d\mu 
\frac{1}{\mu+\Delta} \frac{\delta(z,h)}{g(z)}
\left(\frac{\mu}{\frac{\mu}{\Delta} + \frac{g(z+h)}{g(z)}}\right)
\frac{\delta(z,h)}{g(z)}
\frac{1}{\mu+\Delta}
\\
&= - R(\Delta,h).
\end{aligned}
\end{align}
Together, \eqref{eqn:leftover} and \eqref{eqn:convergent} lead to \eqref{eqn:entcorr}.

Positivity of $R(\Delta,h)$ follows from the fact that the term inside the parentheses in the integrand takes the 
form of a parallel sum
\begin{align}
\frac{1}{\frac{\mu}{\Delta} + \frac{g(z+h)}{g(z)}} = 
\left(\frac{\Delta}{\mu}: \frac{g(z)}{g(z+h)}\right)
\end{align}
which is known to map positive operators to positive operators,
i.e.\ if $A>0$ and $B>0$, then $(A:B)\equiv
(A^{-1}+B^{-1})^{-1}>0$ \cite{simon2019loewner}. Since $\frac{\Delta}{\mu}>0$ and $\frac{g(z)}{g(z+h)}>0$,
the above operator is positive. The terms outside the 
parentheses are just a conjugation of the form $C^\dagger (\cdot) C$,
which preserves positivity. Integrating over $\mu$ and $z$
with respect to a positive measure
also preserves positivity.
\end{proof}

This lemma can then be applied to derive bounds on
the subleading corrections of the entropy of the state $|\tilde{\varphi}_f\rangle$.  These bounds 
can be formulated as follows:
\begin{lemma}\label{lem:entropybounds}
The entropy of the density matrix (\ref{eqn:rhoequiv}) 
satisfies the bounds
\beq \label{eqn:entbounds}
0\leq S(\rho_{\tilde{\varphi}_f}) +S_{\text{rel}}(\varphi\parallel\omega) 
- \vev{\hat{X}-\log g(\hat{X})}_f \leq
-\vev{\int_{-\infty}^\infty dz\, g(z)
\log\left(\frac{g(z-h)}{g(z)}\right)}_{\varphi}.
\eeq
\end{lemma}
\begin{proof}
From equation (\ref{eqn:logrho}) and lemma \ref{lem:Rdelta}, 
the entropy can be written
\begin{align}
S(\rho_{\tilde{\varphi}_f}) = \vev{-\log\rho_{\tilde{\varphi}_f}}_{\tilde{\varphi}_f}
=
-S_{\text{rel}}(\varphi\parallel\omega)+
\vev{\hat{X}-\log g(\hat{X})}_f
 + \vev{-R(\Delta,h)-\int_{-\infty}^\infty dz\, \delta(z,h)
\log g(z)}_\varphi.
\end{align}
The upper bound in (\ref{eqn:entbounds}) follows immediately
from the positivity of $R(\Delta,h)$, after noting that
\beq
\int dz \,\delta(z,h)\log g(z)
=
\int dz(g(z+h)-g(z))\log g(z) = 
\int dz\, g(z)\log\left(\frac{g(z-h)}{g(z)}\right).
\eeq
For the lower bound, we note that because the map $\langle f|\cdot|f
\rangle$ sending an element of the crossed product
algebra $\mathcal{A} \rtimes \mathbb{R}$ to $\mathcal{A}$ is UCP, the dual map sending the state 
$|\varphi\rangle $ to $|\tilde{\varphi}_f\rangle = 
|\varphi\rangle\otimes|f\rangle$ is a quantum channel.  
Relative entropy must decrease 
under the action of a quantum channel \cite{Uhlmann:1976me, Petz1993}, 
so by using the form of the vacuum density
matrix $\rho_{\tilde{\omega}_f} = g(\hat{X})e^{-\hat{X}}$
given in (\ref{eqn:rhovac}), we find that
\begin{align}
S_\text{rel}(\varphi\parallel\omega)\geq &\;
S_\text{rel}(\tilde{\varphi}_f\parallel\tilde{\omega}_f)
\label{eq:Srel} \\
&=
-S(\rho_{\tilde{\varphi}_f}) -\vev{\log\rho_{\tilde{\omega}_f}}_{\tilde{\varphi}_f}
\\
&=
-S(\rho_{\tilde{\varphi}_f}) +\vev{\hat{X}-\log g(\hat{X}}_f .
\end{align}
where, in \eqref{eq:Srel}, the left-hand side is the relative entropy of the states $\varphi$ and $\omega$ for the QFT algebra $\alg$ 
and the right-hand side corresponds to relative entropy of $\tilde{\varphi}_{f}$ and $\tilde{\omega}_{f}$ on the crossed product algebra $\alg\rtimes \mathbb{R}$. Rearranging terms then verifies the lower bound in 
(\ref{eqn:entbounds}).
\end{proof}

We now specialize to a specific choice of the wavefunction
$f(X)$ given by
\beq
f(X) = \sqrt{\frac{\epsilon}{2}} \sech\Big(\epsilon(X-X_0)\Big)
\eeq
For this choice, the integral for the upper bound 
appearing in lemma \ref{lem:entropybounds} evaluates to
\beq
-\epsilon \int_{-\infty}^\infty dz\,\sech^2\Big(
\epsilon(z-h-X_0)\Big) 
\log\left(\frac{\sech(\epsilon(z-h-X_0))}{\sech(\epsilon(z-X_0))}\right)
=
2(\epsilon h\coth(\epsilon h) - 1) \leq 2\epsilon|h|.
\eeq
Hence, the entropy of $\rho_{\tilde{\varphi}_f}$ is upper-bounded
by $2\epsilon\vev{|h|}_\varphi$.  This can be related to the 
relative entropy by decomposing the relative modular Hamiltonian
into the positive and negative components of
its spectrum,
$h= h_+ + h_-$ (note that this split
is unrelated to the split into one-sided
modular Hamiltonians considered in (\ref{eq:modHamsplit})). Then $|h| = h_+-h_-$, and 
\beq
\vev{|h|}_\varphi = S_\text{rel}(\varphi\parallel\omega) +
2\vev{-h_-}_\varphi.
\eeq
To bound the expectation value of $-h_-$, 
we note that the inequality $-h\leq e^{-h}-1$
leads upon
projecting onto the negative part of the $h$ spectrum with
the operator $P_-$ to
$-h_-\leq P_- e^{-h}P_- - P_-\leq e^{-h} = \Delta_{\omega|\varphi}$,
which then implies
\beq
\vev{-h_-}_\varphi \leq \langle\varphi |
\Delta_{\omega|\varphi}|\varphi\rangle
=\langle\omega|\omega\rangle = 1.
\eeq
Combining these results, we see that an upper bound
for the entropy correction displayed in (\ref{eqn:entbounds})
is given by $2\epsilon(S_\text{rel}(\varphi\parallel\omega)+2)$,
which, for states with finite relative entropy,
vanishes as $\epsilon\rightarrow 0$.

Collecting these results, we can now prove the main theorem
of this work.
\begin{theorem}
\label{thm:main}
Let $|\varphi\rangle, |\omega\rangle$ be cyclic-separating vector states for the
von Neumann algebra $\mathcal{A}$ of quantum fields in a causally complete
spacetime subregion.  Assume that there exists a decomposition of the modular
Hamiltonian $h_\omega = h_\omega^{\mathcal{A}}+h_{\omega}^{\mathcal{A}'}$,
where $h_\omega^{\mathcal{A}}$, $h_{\omega}^{\mathcal{A}'}$ are densely-defined
Hermitian forms affiliated with the respective algebras $\mathcal{A}$, $\mathcal{A}'$,
and that $\langle\omega| h_\omega^{\mathcal{A}}|\omega\rangle = 
\langle \omega| h_{\omega}^{\mathcal{A}'}|\omega\rangle = 0$.  Let 
$|\tilde\varphi_\epsilon\rangle$, $|\tilde\omega_\epsilon\rangle$ be the associated 
regulated states for the crossed product algebra, defined by the prescription given above 
equation (\ref{eq:f(X)}).  If $|\varphi\rangle$ is in the domain of $h^{\mathcal{A}}_{\omega}$ and $S_\text{rel}
(\varphi\parallel\omega)$ is finite, then the limit 
\beq \label{eqn:delSvndefn}
\Delta S_{\text{vN}}(\varphi,\omega)=\displaystyle-\lim_{\epsilon\rightarrow 0}
\Tr\left(\rho_{\tilde\varphi_\epsilon} \log\rho_{\tilde\varphi_\epsilon}
-\rho_{\tilde\omega_\epsilon}\log\rho_{\tilde\omega_\epsilon}\right)
\eeq
 exists,
and satisfies
\beq \label{eqn:DSvN}
\Delta S_{\text{vN}}(\varphi,\omega)=\langle\varphi|h_\omega^{\mathcal{A}}|\varphi\rangle
- S_{\text{rel}}(\varphi\parallel\omega).
\eeq
\end{theorem}
\begin{proof}
The regulated entropy $S(\rho_{\tilde\varphi_\epsilon})
=-\Tr\rho_{\tilde\varphi_\epsilon} \log\rho_{\tilde\varphi_\epsilon}$ may be expressed
\beq \label{eqn:Srhoep}
S(\rho_{\tilde\varphi_\epsilon}) = -S_\text{rel}(\varphi\parallel\omega)
+\vev{\hat{X}-\log g(\hat{X})}_{\tilde\varphi_\epsilon} + \epsilon\mathcal{E}
\eeq
where, by lemma \ref{lem:entropybounds}, 
the error term $\mathcal{E}$ is bounded between $0$ and 
the right hand side of (\ref{eqn:entbounds}).
Using the regulated probability distribution
\beq
g(X) = \frac{\epsilon}{2}\sech^2\left(X-\vev{h_\omega^{\mathcal{A}}}_\varphi\right),
\eeq
which is chosen according to the prescription described
in equation (\ref{eq:f(X)}), the upper bound for $\mathcal{E}$ can be
taken to be $2(S_\text{rel}(\varphi\parallel\omega) +2)$ 
by the above discussion.  

This choice of wavefunction then gives
\beq
\vev{\hat{X}-\log g(\hat{X})}_{\tilde\varphi_\epsilon} = \vev{h_\omega^{\mathcal{A}}}_\varphi
-\log(2\epsilon)+2.
\eeq
On the other hand, using the density matrix 
(\ref{eqn:rhovac}), the vacuum entropy $S(\rho_{\tilde\omega_\epsilon}) = -\Tr \rho_{\tilde\omega_\epsilon}\log\rho_{\tilde\omega_\epsilon}$ evaluates exactly to
\beq\label{eqn:Svac}
S(\rho_{\tilde\omega_\epsilon}) = -\log(2\epsilon) + 2.
\eeq
Hence, we see that the divergent $\log(2\epsilon)$ term cancels out in the entropy 
difference, and the term $\epsilon \mathcal{E}$ in (\ref{eqn:Srhoep}) vanishes as $\epsilon\rightarrow 0$.
Therefore, the entropy difference $S(\rho_{\tilde\varphi_\epsilon}) 
- S(\rho_{\tilde\omega_\epsilon})$ has a finite limit as $\epsilon\rightarrow 0$,
and is given by (\ref{eqn:DSvN}).
\end{proof}

Although the assumption that the state $|\varphi\rangle$ was cyclic-separating
was necessary in the above proof of theorem \ref{thm:main}, it seems likely
that this requirement can be lifted to consider general states 
$|\varphi\rangle$.  The main subtleties involve determining the density
matrix on the crossed product algebra without assuming that $|\varphi\rangle$
is cyclic-separating, and then checking that the errors 
in the entropy expresssion can still be bounded.  We leave this 
as an interesting question for future work. 

\section{Proof of the Bekenstein Bound}
\label{sec:Bekbound}
If a spacetime has a complete time-like Killing vector field $\xi$, a state $\omega$ is called \textit{globally KMS} if it obeys the KMS condition \cite{Kubo:1957mj,MS_1959} with respect to the flow generated by the Killing field. In such a case, assuming the existence of stress tensor $T_{\mu \nu}$, the modular Hamiltonian of $\ket{\omega}$ is
\begin{align}
    h_{\omega} =\frac{2\pi }{\kappa} \int_{\Sigma} \sqrt{h}d^{d-1}x ~n^{\mu}\xi^{\nu}T_{\mu \nu}
\end{align}
where $\Sigma$ is a complete Cauchy surface for the spacetime, $\sqrt{h}$ is the volume form, and $n^{\mu}$ is the unit normal to $\Sigma$. Three well-known examples of globally KMS states (see e.g., \cite{Hollands:2014eia}) are the Minkowski vacuum restricted to the Rindler wedge, the Bunch-Davies state restricted to the static patch of de Sitter, and the Hartle-Hawking state restricted to the exterior of the black hole.

\begin{theorem}[Bekenstein Bound]
\label{thm:Bekbound}
Let $|\varphi\rangle, |\omega\rangle$ be 
vector states that are cyclic and separating for the
von Neumann algebra $\mathcal{A}$ of quantum fields in a causally complete
spacetime subregion $\mathcal{R}$. Let $\omega$ be globally KMS. Assuming the existence of a stress tensor $T_{\mu \nu}$ and the finiteness of the relative entropy $S_{\textrm{rel}}(\varphi,\omega)$, $\ket{\varphi}$ obeys the following bound
  \begin{align}
  \Delta S_{\text{vN}}(\varphi,\omega)\leq  \frac{2\pi }{\kappa}\int_{\Sigma_{\mathcal{R}}}\sqrt{h}d^{d-1}x n^{\mu}\xi^{\nu}\bra{\varphi}T_{\mu \nu}\ket{\varphi}.
  \end{align}
\end{theorem}
\begin{proof}
    The modular Hamiltonian is seen to decompose into Hermitian forms 
\begin{align}
    h_{\omega}^{\mathcal{A}} = \frac{2\pi }{\kappa} \int_{\Sigma_{\mathcal{R}}} \sqrt{h}d^{d-1}x ~n^{\mu}\xi^{\nu}T_{\mu \nu},\quad h_{\omega}^{\mathcal{A}'} = \frac{2\pi }{\kappa}\int_{\Sigma_{\mathcal{R}^{\prime}}} \sqrt{h}d^{d-1}x ~n^{\mu}\xi^{\nu}T_{\mu \nu},
\end{align}
where $\mathcal{R}'$ is the complement of $\mathcal{R}$ on $\Sigma$. From section \ref{subsec:sesq}, we know these forms are densely-defined and affiliated with $\mathcal{A}$ and $\mathcal{A}
'$ respectively. Moreover, $n^{\mu}\bra{\omega}T_{\mu \nu}(x)\ket{\omega}\xi^{\nu} = 0$ for all $x$ so $\bra{\omega}h_{\omega}^{\mathcal{A}}\ket{\omega}=\bra{\omega}h_{\omega}^{\mathcal{A}'}\ket{\omega}= 0$. We may then invoke theorem \ref{thm:main} to find
\beq \
\Delta S_{\text{vN}}(\varphi,\omega)=\langle\varphi|h_\omega^{\mathcal{A}}|\varphi\rangle
- S_{\text{rel}}(\varphi\parallel\omega) = \frac{2\pi }{\kappa} \int_{\Sigma_{\mathcal{R}}} \sqrt{h}d^{d-1}x ~n^{\mu}\xi^{\nu}\bra{\varphi}T_{\mu \nu}\ket{\varphi}
- S_{\text{rel}}(\varphi\parallel\omega).
\eeq
The relative entropy is strictly positive for distinct states and zero for identical states \cite{araki1976relative}, so
\beq \
\Delta S_{\text{vN}}(\varphi,\omega)\leq  \frac{2\pi }{\kappa} \int_{\Sigma_{\mathcal{R}}} \sqrt{h}d^{d-1}x ~n^{\mu}\xi^{\nu}\bra{\varphi}T_{\mu \nu}\ket{\varphi}
.
\eeq
\end{proof}

\section{Proof of the Quantum Null Energy Condition}
Here we fill in the details for demonstrating that the 
entropy difference $\Delta S_{\text{vN}} (\varphi,\omega)$
satisfies the QNEC inequality (\ref{eq:QNEC}).  
As shown by Casini, Teste, and Torroba 
\cite{2017JPhA...50J4001C}, the algebra
$\alg_{V(y)}$ associated with cuts of the Rindler horizon
forms a half-sided modular inclusion with respect to the $v=0$
horizon cut in the Minkowski vacuum $|\omega\rangle$.
This fact can be used to write the modular Hamiltonian of 
the state $|\omega\rangle$ for the algebra $\alg_{V(y)}$ 
as a local integral of the stress tensor, which 
naturally splits into Hermitian forms $h_\omega^{\alg_V}$
and $h_\omega^{\alg_V'}$, given explicitly by integrals
over the Rindler horizon,
\begin{align}
h_\omega^{\alg_V} 
&= 
2\pi \int_{\mathbb{R}^{d-2}} d^{d-2} y \int_{V(y)}^\infty
dv (v-V(y))T_{vv}(v,y) 
\label{eqn:hAV}
\\
h_\omega^{\alg_V'} 
&= 
2\pi \int_{\mathbb{R}^{d-2}} d^{d-2} y \int_{-\infty}^{V(y)}
dv (v-V(y))T_{vv}(v,y).
\end{align}

The formula for the entropy difference involves the expectation
value of $h_{\omega}^{\alg_V}$ in a state $|\varphi\rangle$.  
Since the QNEC involves second variational derivatives of
the entropy, we will need to take functional derivatives
of $\langle \varphi| h_\omega^{\alg_V}|\varphi\rangle$.  
These can be computed from the expression (\ref{eqn:hAV})
as follows
\begin{align} \label{eqn:del2hAV}
\frac{\delta^2}{\delta V(y)^2}\langle\varphi|h_\omega^{\alg_V}
|\varphi\rangle
&=
\frac{\delta}{\delta V(y)}\langle \varphi|
-2\pi \int_{V(y)}^\infty dv T_{vv}(v,y) |\varphi\rangle
=
2\pi\langle \varphi| T_{vv}(V(y), y)|\varphi\rangle.
\end{align}

We now invoke the version of the QNEC proved by Ceyhan and Faulkner
in terms of the relative entropy of the state $|\varphi\rangle$
with respect to the vacuum  \cite{2018arXiv181204683C}.  
Their results imply that the 
relative entropy $S_\text{rel}(\varphi\parallel\omega;\alg_{V})$
is a convex functional of the horizon cut $V(y)$,
i.e.\
that the second variational derivative is nonnegative,
\beq \label{eqn:del2Srel}
\frac{\delta^2}{\delta V(y)^2} 
S_{\text{rel}}(\varphi\parallel \omega;\alg_V) \geq 0.
\eeq
This inequality is rigorously established in 
\cite{2018arXiv181204683C}
under the assumptions that the average null energy and 
relative entropy for the the state $|\varphi\rangle$ are 
finite:
\beq \label{eqn:conditions}
\langle \varphi |
\int_{-\infty}^\infty dv T_{vv}(v,y)|\varphi\rangle < \infty;
\qquad
S_{\text{rel}}(\varphi\parallel\omega;\alg_{V}) <\infty.
\eeq
Combining these results, we can demonstrate that the 
regulated entropies satisfy the QNEC:

\begin{theorem}[Quantum Null Energy Condition]
\label{thm:QNEC}
Let $|\varphi\rangle$ be a cyclic and separating state with finite 
average null energy and relative entropy, as defined by 
the conditions (\ref{eqn:conditions}). Then the local
stress energy expectation value is bounded below
by the second variation of $\Delta S_{\text{vN}}(\varphi,\omega;
\alg_V)$, where $V(y)$ is a cut of a Rindler horizon,
\beq
2\pi \langle \varphi |T_{vv}(V(y),y)|\varphi\rangle
\geq
\frac{\delta^2}{\delta V(y)^2} \Delta S_{vN}(\varphi,\omega,
\alg_V).
\eeq
Additionally, assuming that the second variational derivative
of the error term $\frac{\delta^{2}}{\delta V(y)^2}\mathcal{E}$
in the regulated entropy 
defined in (\ref{eqn:Srhoep}) is uniformly bounded as $\epsilon\rightarrow 0$, 
the QNEC holds directly
for the regulated entropy $S_{\textrm{vN}}(\rho_{\tilde\varphi_\epsilon})$,
\beq
2\pi \langle \varphi |T_{vv}(V(y),y)|\varphi\rangle
\geq
\lim_{\epsilon\rightarrow 0}
\frac{\delta^2}{\delta V(y)^2} S_{\textrm{vN}}(\rho_{\tilde\varphi_\epsilon})
\eeq
without the need to subtract the vacuum entropy 
$S_{\textrm{vN}}(\rho_{\tilde\omega_\epsilon})$.
\end{theorem}
\begin{proof}
By Theorem \ref{thm:main}, the entanglement entropy difference
satisfies
\beq
\Delta S_{\text{vN}}(\varphi,\omega;\alg_V)
=
\langle \varphi|h_\omega^{\alg_V}|\varphi\rangle
-
S_{\text{rel}}(\varphi||\omega; \alg_V).
\eeq
Taking second variational derivatives of this equality and 
invoking (\ref{eqn:del2hAV}) and (\ref{eqn:del2Srel}), 
we find that 
\beq
\frac{\delta^2}{\delta V(y)^2} 
\Delta S_{\text{vN}}(\varphi,\omega; \alg_V)
\leq
\frac{\delta^2}{\delta V(y)^2}
\langle\varphi|h_{\omega}^{\alg_V}|\varphi\rangle
=
2\pi\langle\varphi |T_{vv}(V(y),y)|\varphi\rangle
\eeq
as desired.

It is important in this argument that the limit 
in the definition of the regulated entropy (\ref{eqn:delSvndefn})
is taken {\it before} computing the variational derivatives 
$\frac{\delta^2}{\delta V(y)^2}$.  However, in cases where 
the limit
commutes with the derivatives, switching their order allows the 
dependence on the vacuum entropy to be dropped from the QNEC
inequality.  This is because the regulated vacuum entropy
is independent of the cut, i.e.
\beq
S_{\textrm{vN}}(\rho_{\tilde\omega_\epsilon};\alg_V) = - \log (2\epsilon)+2 
\eeq
for the case of the $\sech$ regulator employed 
in Theorem (\ref{thm:main}) (see equation (\ref{eqn:Svac})).
Hence, the second variation of the regulated entropy
satisfies
\begin{align}
\frac{\delta^2}{\delta V(y)^2} S_{\textrm{vN}}(\rho_{\tilde\varphi_\epsilon};
\alg_V) 
&= 
\frac{\delta^2}{\delta V(y)^2}
\Big( S_{\textrm{vN}}(\rho_{\tilde\omega_\epsilon},\alg_V) - 
S_{\text{rel}}(\varphi\parallel\omega;\alg_V)
+\langle\varphi|h_\omega^{\alg_V}|\varphi\rangle
+ \epsilon \mathcal{E}\Big)
\\
&=
\frac{\delta^2}{\delta V(y)^2}\Big(  - 
S_{\text{rel}}(\varphi\parallel\omega;\alg_V)
+\langle\varphi|h_\omega^{\alg_V}|\varphi\rangle
\Big)
+\epsilon \frac{\delta^2}{\delta V(y)^2}\mathcal{E}\Big).
\end{align}
From the assumption of uniform boundedness of $\frac{\delta^2}
{\delta V(y)^2} \mathcal{E}$ as $\epsilon\rightarrow 0$, the final term vanishes in the limit as 
$\epsilon\rightarrow0$.  In this case, we find that 
\beq
\lim_{\epsilon\rightarrow 0} \frac{\delta^2}{\delta V(y)^2}
S_{\textrm{vN}}(\rho_{\tilde\varphi_\epsilon},\alg_V) = 
\frac{\delta^2}{\delta V(y)^2}\Delta S_{\text{vN}}(\varphi,
\omega;\alg_V),
\eeq
and hence the former quantity can be used as the lower
bound in the QNEC.
\end{proof}

As a final comment, we note that while the assumption that $\frac{\delta^2}{\delta V(y)^2}\mathcal{E}$ is finite 
does not appear to directly follow from the arguments
of section  \ref{app:bound}, it seems likely that this quantity is 
finite in a wide class of states.  An interesting direction
for future work would be to determine what additional
assumptions are needed to ensure that this error term
has finite variational derivatives.  Note also that Ceyhan and Faulkner's 
proof of convexity of the relative entropy does not assume that 
$|\varphi\rangle$ is cyclic-separating, 
so this assumption in theorem \ref{thm:QNEC} could be lifted if
theorem \ref{thm:main} can be generalized to this case.

\end{document}